\documentclass[twocolumn,10pt,journal]{IEEEtran}
\usepackage{indentfirst,amsmath,dblfloatfix}
\usepackage{color,rotating,pdflscape}
\usepackage{cite}
\usepackage{amsfonts,amsmath,amssymb,amsthm}
\usepackage{latexsym,amscd}
\usepackage{amsbsy,extarrows}
\usepackage{graphicx,multirow,bm}
\usepackage[table]{xcolor}

\usepackage{tikz}
\usepackage{url} 
\usetikzlibrary{arrows,shapes}
\usetikzlibrary{shadows}

\newtheorem{Theorem}{Theorem}

\newtheorem{Definition}{Definition}

\newtheorem{Remark}{Remark}
\newtheorem{Lemma}{Lemma}

\newtheorem{Construction}{Construction}

\begin{document}
\title{Efficient Repair of $(k+2, k)$ Degraded Read Friendly MDS Array Codes With\\ Sub-packetization $2$ }

\author{Jie~Li,~\IEEEmembership{Senior~Member,~IEEE,} and 
Xiaohu~Tang,~\IEEEmembership{Fellow,~IEEE}

\thanks{The work of X. Tang was supported by the National
Natural Science Foundation of China under Grant 62371401.}
\thanks{J. Li was with the Information Coding and Transmission Key Lab of Sichuan Province, CSNMT Int. Coop. Res. Centre (MoST), Southwest Jiaotong University, Chengdu, 610031, China, and is with the Department of Computer Science, City University of Hong Kong, Hong Kong SAR 
  (e-mail: jieli873@gmail.com).}

\thanks{X. Tang is with the Information Coding and Transmission Key Lab of Sichuan Province, CSNMT Int. Coop. Res. Centre (MoST), Southwest Jiaotong University, Chengdu, 610031, China (e-mail: xhutang@swjtu.edu.cn).}
}
\maketitle

\begin{abstract}
In this paper, we present two constructions of degraded read friendly (DRF) MDS array codes with two parity nodes and a sub-packetization level of $2$ over small finite fields, applicable to any arbitrary code length. The first construction achieves the smallest repair bandwidth among all existing constructions with the same parameters, and is asymptotically optimal with respect to the lower bound on the average repair bandwidth characterized by Zhang \textit{et al.} The second construction supports two repair mechanisms, depending on whether computation within the helper nodes is permitted or not during the node repair process, thereby optimizing either the repair bandwidth or the rebuilding access. 
\end{abstract}

\section{Introduction}

Distributed storage systems have widespread practical applications, such as Windows Azure Storage \cite{calder2011windows} and Hadoop-based systems deployed in Facebook and Yahoo \cite{borthakur2007hadoop}. In these applications, data are stored across multiple unreliable storage nodes, making node failures common rather than a rare exception. To ensure reliability, erasure codes have been widely deployed in distributed storage systems, e.g., HDFS RAID \cite{hdfs}.
As an important class of erasure codes, MDS codes provide the optimal tradeoff between fault tolerance and storage overhead.
However, MDS codes are inefficient in the context of node repair. For instance, in a distributed storage system based on an $(n,k)$ MDS code, repairing a failed node requires downloading the entire content from any $k$ surviving nodes. This results in an excessive \textit{repair bandwidth}, which is defined as the amount of data downloaded to repair a failed node.

To enable efficient node repair, one can employ MDS array codes, where a codeword is represented by an $N\times n$ array for $N\ge1$. In this context, $N$ is referred to as the \textit{sub-packetization level}. According to the MDS property and the cut-set bound demonstrated in \cite{dimakis2010network}, the repair bandwidth of $(n,k)$ MDS array codes with sub-packetization level $N$ is lower bounded by
\begin{equation}\label{Eqn_gamma_opt_d}
\gamma\ge\gamma_{\rm optimal}=
\left\{
\begin{array}{ll}
k, &\mbox{\ if\ } N=1, \\[2pt]
\frac{d}{d-k+1}N, & \mbox{\ if\ }N>1.
\end{array}
\right.
\end{equation}
Here, $d$ such that $k\le d<n$ denotes the number of helper nodes contacted during the repair process. MDS array codes ($N>1$) that achieve this lower bound are said to have the \textit{optimal repair bandwidth}.
In \cite{dimakis2010network}, they are also referred to as MSR codes.

Besides the repair bandwidth, \textit{rebuilding access} or \textit{disk I/O cost} is also a critical metric for evaluating the performance of an MDS array code, which is defined as the amount of data accessed from the helper nodes during a repair process. Notably, the rebuilding access of an MDS array code is at least as large as its repair bandwidth. There are MSR codes where the rebuilding access also meets the lower bound of the repair bandwidth $\gamma_{\rm optimal}$, e.g., those in \cite{li2018generic,li2017generic,ye2017explicitB,ye2017explicit,tamo2012zigzag}.

Over the past decade, MSR codes have garnered significant attention \cite{rashmi2011optimal,tamo2012zigzag,papailiopoulos2013repair,tian2014characterizing,han2015update,li2015framework,sasidharan2016explicit,wang2016explicit,li2016optimal,ye2017explicit,ye2017explicitB,li2017generic,goparaju2017minimum,li2018generic,balaji2018erasure,chen2019explicit,hou2019binary,elyasi2020cascade,liu2022generic,li2024msr,zhu2025two}. 
However, MSR codes suffer from a large sub-packetization level $N$,
e.g., $N\ge r^{\frac{n}{r+1}}$ when $d=n-1$, where $r=n-k$. This is necessary as Alrabiah and Guruswami recently showed that the lower
bound on $N$ is exponential in $k$ and conjectured
that the construction in \cite{wang2016explicit} with $N=r^{\frac{n}{r+1}}$ is exactly tight \cite{alrabiah2021exponential}. The constraint affects the selection of system parameters and complicates metadata management.  Specifically, $n$ and $r$ should be kept small to avoid a large sub-packetization level, which hinders implementation in practical systems \cite{rawat2018mds}.  

In this paper, we consider MDS array codes that achieve both small repair bandwidth/rebuilding access and the smallest sub-packetization level. 
To reduce repair bandwidth/rebuilding access, the smallest sub-packetization level is $2$ as implied by \eqref{Eqn_gamma_opt_d}. Moreover, the highest code-rate is $\frac{k}{k+2}$. These aspects are the focus of this work, and the lower bound in \eqref{Eqn_gamma_opt_d} reduces to $\gamma\ge n-1$ (which subsequent research has shown to be unachievable). An advantage of MDS array codes with a sub-packetization level of $2$ is that discontinuous disk I/O  can be effectively avoided during the repair process. Codes with such a focus were also studied in \cite{ma2024repair} and \cite{liu2024family} recently. 

In the literature, there are a few constructions of $(n=k+2, k)$ MDS array codes with an efficient repair strategy and a sub-packetization level of $2$. The MSR code construction in \cite{rashmi2011optimal} supports code parameters with two parity nodes and a sub-packetization level of $2$, but the code length is limited to $n=5$, and the required finite field size must satisfy $q>10$. In \cite{wang2016explicit}, an example of $(n=6, k=4)$ MDS array code with a sub-packetization level of $2$ was presented, however, the two parity nodes do not achieve optimal repair bandwidth, resulting in a higher average repair bandwidth than that in \cite{rashmi2011optimal}. In \cite{guan2017construction}, an $(n=5,k=3)$ MSR code with a sub-packetization level of $2$ was proposed, which can be built over a finite field with size $q\ge 4$, i.e., much smaller than the one in \cite{rashmi2011optimal} while the remaining properties are the same. This construction was later generalized to $n=6$ in \cite{guan2017new}.
For convenience, we refer to the four constructions in \cite{rashmi2011optimal,wang2016explicit,guan2017construction,guan2017new} as the product-matrix MSR code, the long MDS code, the GKW code, and the GKWX code, respectively. 

{In practice, during the interval between failure and recovery, users may request partial data that resides on a failed node. To access such temporarily unavailable data, a storage system performs a \textit{degraded read}, which reconstructs the missing data from surviving nodes. 
Degraded reads occur more frequently than node repair operations and often operate under stricter latency constraints. Hence, optimizing degraded read performance is critical \cite{ford2010availability,khan2012rethinking,zhu2014boosting,wu2021achievable}.
We formally define a degraded read as follows.}
\begin{Definition}[Degraded Read \cite{wu2021achievable}]
When a user requests partial data that is temporarily unavailable, the read is called a degraded read because the data must be reconstructed from other surviving nodes.
\end{Definition}

To enhance users' experience by providing faster access to data, the degraded read friendly (DRF) property is quite essential in practical systems, as it ensures the system's availability.  In this paper, we define the 
DRF property as follows.

\begin{Definition}[Degraded Read Friendly \cite{wu2021achievable}]\label{Def_DRF}
If for any node $i$ and any symbol index $t \in [0: N)$, there exists a decoding strategy for the symbol $f_{i,t}$ such that 
\begin{itemize}
    \item [i)] No multiple accessed symbols in surviving nodes, and
    \item [ii)] No multiplication is needed during the degraded-read reconstruction,
\end{itemize}
then the array code is said to be degraded read friendly.
\end{Definition}

Constructions of MDS codes enabling efficient degraded read to multiple symbols have also been studied in \cite{khan2012rethinking,zhu2014boosting}. In this work, we focus on the degraded read to a single symbol in MDS array codes.
In \cite{wu2021achievable}, a lower bound for the average rebuilding access of $(n=k+2, k)$ DRF MDS array codes with a sub-packetization level of $2$ was derived for general $n$. A construction matching this lower bound was also proposed, which will be referred to as the WHLBZZW code in this paper. Very recently, in \cite{li2026mds}, a DRF code named LDIO was proposed, which aims to optimize the rebuilding access of only the systematic nodes. It works for general code parameters as long as the number $r$ of parity nodes and the sub-packetization level $N$ satisfy $(N-1)\mid (r-1)$.
The results in \cite{wu2021achievable} and recent works \cite{rawat2018mds,li2021systematic,li2023mds,li2023mds-tit,li2026mds} demonstrate that MDS array codes with small sub-packetization levels and large code lengths can be constructed by sacrificing the optimality of the repair bandwidth and rebuilding access w.r.t. the lower bound in \eqref{Eqn_gamma_opt_d}.
However, lower bounds of the average repair bandwidth of $(n=k+2, k)$ DRF MDS array codes with a sub-packetization level of $2$ remain open.

Recall that rebuilding access is also a critical metric in distributed storage. Unlike MSR codes, the lower bound on rebuilding access is not the same as that on the repair bandwidth for general MDS array codes. In a recent work \cite{zhang2025optimal}, Zhang \textit{et al.} studied $(n=k+2,k)$ MDS array codes with sub-packetization $2$ without the DRF property.
They showed that the lower bound on repair bandwidth is smaller than that on rebuilding access.

Since both repair bandwidth and rebuilding access are important metrics, in this paper we focus on optimizing the two metrics separately.  Specifically,
we propose two new $(n=k+2, k)$ DRF MDS array codes with a sub-packetization level of $2$. 
The first code aims to optimize repair bandwidth rather than rebuilding access.
Its average repair bandwidth is significantly smaller than the lower bound on average rebuilding access derived in \cite{wu2021achievable} for the same parameters. The second MDS array code supports two repair mechanisms, one optimizing rebuilding access while the other optimizing repair bandwidth. One can choose the repair mechanism based on the computing capability of the helper nodes during the repair process. 
As in some distributed storage systems, such as Ceph \cite{Ceph} and Hadoop \cite{porter2010decoupling,HadoopCluster}, storage nodes are capable of performing computational tasks. Ideally, helper nodes should avoid computation during node repairs when they are already engaged in these tasks. 

To summarize, this paper contains the following contributions:

\begin{itemize}
    \item [i)] We propose an $(n=4m, k=4m-2)$ MDS array code with sub-packetization level $2$ over $\mathbf{F}_q$, where $q=2^{2t}\ge \frac{3n}{4}+1$. 
    The normalized repair bandwidth (repair bandwidth divided by file size) of this code is $\frac{5n-8}{8n-16}$. 
    This value is equal to the normalized repair bandwidth of the one derived by Zhang \textit{et al.} \cite{zhang2025optimal}. Apart from this, it is smaller than that of all other existing $(n=k+2,k)$ MDS array codes with sub-packetization $2$. 
    Moreover, it is asymptotically optimal with respect to the lower bound derived in \cite{zhang2025optimal}.
    
    \item [ii)] We propose an $(n=3m, k=3m-2)$ MDS array code with sub-packetization level $2$ over $\mathbf{F}_q$. 
    The finite field size satisfies $q\ge \frac{2n}{3}$ if $2\mid q$ and $q\ge n$ otherwise. 
    This code supports two distinct repair mechanisms:
    \begin{itemize}
        \item The first mechanism achieves an average normalized repair bandwidth of $\frac{2n-3}{3n-6}$, which is lower than that achieved by the construction in \cite{wu2021achievable}.
        \item The second mechanism achieves an average normalized rebuilding access of $\frac{13n-18}{18n-36}$, which is very close to the lower bound in derived \cite{wu2021achievable}.
    \end{itemize}
\end{itemize}

The rest of this paper is organized as follows. Section \ref{sec:pre} introduces some preliminaries on MDS array codes. Section \ref{sec:new} proposes the new construction of $(n=k+2, k)$ MDS array codes with sub-packetization level $2$ and small repair bandwidth. Section \ref{sec:C2} presents another construction of $(n=k+2, k)$ MDS array codes with multiple repair strategies. An extensive comparison is carried out in Section \ref{sec:comp}. Finally, Section \ref{sec:concl} concludes this paper.

\section{Preliminaries}\label{sec:pre}
In this section, we introduce some preliminaries on high-rate MDS codes. Let $q$ be a prime power and $\mathbf{F}_q$ the finite field with $q$ elements. For any two integers $a$ and $b$ with $b>a$, denote by $[a: b)$ the set $\{a, a+1, \ldots, b-1\}$.

\subsection{$(n=k+2,k)$ MDS codes}
Let $\mathbf{f}_0, \mathbf{f}_1, \ldots, \mathbf{f}_{n-1}$ represent the data stored across a distributed storage system consisting of $n$ nodes, based on an $(n=k+2,k)$ MDS array code with a sub-packetization level of $2$. Here $\mathbf{f}_i=\begin{pmatrix}
f_{i,0}\\
f_{i,1}
\end{pmatrix}$ is a column vector of length $2$ over $\mathbf{F}_q$, and $f_{i,t}$ is called a \textit{sub-packet} for $t=0,1$.

In this paper, similar to that in \cite{wu2021achievable}, we consider $(n=k+2,k)$ DRF MDS codes that admit the following parity-check form:
\begin{equation}\label{Eqn parity check eq}
\underbrace{\left(
\begin{array}{cccc}
I & I & \cdots & I \\
A_{0} & A_{1} & \cdots & A_{n-1}
\end{array}
\right)}_{A}
\left(\begin{array}{c}
\mathbf{f}_0\\
\mathbf{f}_1\\
\vdots\\
\mathbf{f}_{n-1}
\end{array}
\right)=\mathbf{0}_{4},
\end{equation}
where $A_i$ ($i\in [0: n)$) is a $2\times 2$ matrix, $\mathbf{0}_{4}$ denotes the zero column vector of length $4$, and will be abbreviated as $\mathbf{0}$ in the sequel if its length is clear. The $4\times 2n$ matrix $A$ in \eqref{Eqn parity check eq} is referred to as the \textit{parity-check matrix} of the $(n=k+2,k)$ MDS code with a sub-packetization level of $2$.

From \eqref{Eqn parity check eq}, we have that
\begin{equation*}
    \begin{pmatrix}
f_{0,0}\\
f_{0,1}
\end{pmatrix}+\begin{pmatrix}
f_{1,0}\\
f_{1,1}
\end{pmatrix}+\cdots+\begin{pmatrix}
f_{n-1,0}\\
f_{n-1,1}
\end{pmatrix}=0.
\end{equation*}
That is, for any given $i\in [0: n)$ and $t=0, 1$, a degraded read to the symbol $f_{i,t}$ 
can be performed efficiently by accessing just one symbol $f_{j,t}$ from node $j$ for $0\le j\ne i<n$ and computing
$-\sum\limits_{j=0,j\ne i}^{n-1}f_{j,t}$, meeting the conditions in Definition \ref{Def_DRF}, which indicates that MDS codes defined by \eqref{Eqn parity check eq} are indeed DRF.



An $(n=k+2,k)$ MDS array code defined by \eqref{Eqn parity check eq} possesses the MDS property that the source file can be reconstructed by connecting to any $k$ out of the $n$ nodes. That is, any $2\times 2$ sub-block matrix of $A$
is nonsingular, i.e, $A_i-A_j$ is nonsingular for $i,j\in [0: n)$ with $i\ne j$ \cite{ye2017explicit}.

To end this subsection, we present a lemma that will be used to verify the MDS property of the new array code in the next section.

\begin{Lemma}\label{lem}
Let $q=2^{2t}$ for some positive integer $t$, and let $w$ be a primitive element of the finite field $\mathbf{F}_q$. Then we have
\begin{equation}
w^{2i}+w^{i}+1\ne0 ~\mbox{for}~ 0\le i<\frac{q-1}{3}.
\end{equation}
\end{Lemma}
\begin{proof}
For $0\le i<\frac{q-1}{3}$, it follows that $0\le 3i<q-1$, thus, $w^{3i}\ne 1$. Since $(w^i+1)(w^{2i}+w^{i}+1)=w^{3i}+1\ne0$, then
$w^{2i}+w^{i}+1\ne0$.
\end{proof}

\subsection{Repair Process}\label{sec:repair}
Consider the repair of a failed node $i$ ($i\in[0,n)$) of an $(n=k+2,k)$ MDS array code. Let $R_i$ be a $2\times 4$ matrix. Multiply \eqref{Eqn parity check eq} on the left by $R_i$, we obtain
\begin{equation}\label{Eqn repair eq}
R_i\left(\hspace{-2mm}
\begin{array}{cccc}
I & I & \cdots & I \\
A_{0} & A_{1} & \cdots & A_{n-1}
\end{array}
\hspace{-2mm}\right)
\left(\hspace{-2mm}\begin{array}{c}
\mathbf{f}_0\\
\mathbf{f}_1\\
\vdots\\
\mathbf{f}_{n-1}
\end{array}
\hspace{-2mm}\right)=\mathbf{0}.
\end{equation}
Node $i$ is then regenerated by solving the above equations. Rewrite \eqref{Eqn repair eq} as
\begin{equation}\label{Eqn parity check eq2}
\underbrace{R_i\left(
\begin{array}{c}
I \\
A_{i}
\end{array}
\right)
\mathbf{f}_i}_{\mathrm{useful ~data}}
+\sum_{j=0,j\ne i}^{n-1}\underbrace{R_i\left(\hspace{-2mm}
\begin{array}{c}
I \\
A_{j}
\end{array}
\hspace{-2mm}\right)
\mathbf{f}_j}_{\mathrm{interference ~by~}\mathbf{f}_j}=\mathbf{0}.
\end{equation}
Clearly, regenerating node $i$ requires the coefficient matrix of the useful data in \eqref{Eqn parity check eq2} to be of full rank, i.e.,
\begin{equation}\label{repair_node_requirement1 n-1}
\textrm{rank}\left(R_i\left(
\begin{array}{c}
I \\
A_{i}
\end{array}
\right)\right) =2, \, i\in [0,n),
\end{equation}
and the interference can be cancelled, which is done by downloading a sufficient amount of data from the surviving nodes.
The amount of data that needs to be downloaded (i.e., the repair bandwidth of node $i$) to cancel the interference in \eqref{Eqn parity check eq2} is given by
\begin{equation}\label{Eqn_RB}
\gamma_i=\sum_{j=0,j\ne i}^{n-1}\mbox{rank}\left(R_i\left(
\begin{array}{c}
I \\
A_{j}
\end{array}
\right ) \right).
\end{equation}
Meanwhile, the rebuilding access (or the disk I/O) of node $i$, i.e., the amount of data that needs to be accessed to cancel the interference in \eqref{Eqn parity check eq2} is
\begin{equation}\label{Eqn_RA}
\Gamma_i=\sum_{j=0,j\ne i}^{n-1}N_c\left(R_i\left(
\begin{array}{c}
I \\
A_{j}
\end{array}
\right)\right),
\end{equation}
where $N_c(A)$ denotes the number of nonzero columns of the matrix $A$.

For example, suppose that
\begin{equation}\label{Eqn_interference}
R_i\left(
\begin{array}{c}
I \\
A_{j}
\end{array}
\right)=\begin{pmatrix}
1&1\\
x&x
\end{pmatrix},
\end{equation}
for some $x$ and $\mathbf{f}_j=\begin{pmatrix}
f_{j,0}\\
f_{j,1}
\end{pmatrix}$. To cancel the interference
\begin{equation*}
R_i\left(
\begin{array}{c}
I \\
A_{j}
\end{array}
\right)\mathbf{f}_j,
\end{equation*}
one needs to access two symbols $ f_{j,0}$ and $f_{j,1}$ from node $j$, but only needs to download one symbol $f_{j,0}+f_{j,1}$, which is consistent with
\begin{equation*}
N_c\left(\begin{pmatrix}
1&1\\
x&x
\end{pmatrix}\right)=2 \mbox{~and~} \mbox{rank}\left(\begin{pmatrix}
1&1\\
x&x
\end{pmatrix}\right)=1.
\end{equation*}

\subsection{Bounds on the Repair Bandwidth and Rebuilding Access of $(k+2, k)$ MDS Array Codes}

To facilitate performance comparisons, we define the normalized repair bandwidth of node $i$ of an $(n, k)$ array code with sub-packetization $N$ as the ratio of the repair bandwidth to the file size, i.e.,
\begin{equation*}
\gamma_i^{Nor}=\frac{\gamma_i}{kN}.
\end{equation*}
We further define the \textit{average normalized repair bandwidth} as
\begin{equation*}
\gamma^{Nor}_{Ave}=\frac{\sum\limits_{i=0}^{n-1}\gamma_i^{Nor}}{n}=\frac{\sum\limits_{i=0}^{n-1}\gamma_i}{nkN}.
\end{equation*}
Similarly, the \textit{average normalized rebuilding access} is defined as
\begin{equation*}
\Gamma^{Nor}_{Ave}=\frac{\sum\limits_{i=0}^{n-1}\Gamma_i}{nkN}.
\end{equation*}

In \cite{wu2021achievable}, a lower bound for the average normalized rebuilding access of $(n=k+2, k)$ DRF MDS array codes with a sub-packetization level of $2$ is derived, showing that
\begin{equation}\label{Eqn_Acess_bound}
\Gamma^{Nor}_{Ave}\ge \frac{\min\{\Delta_3, \Delta_4\}}{2nk},   
\end{equation}
where
\begin{equation*}
\Delta_3=\min_{l_1+l_2+l_3=n, l_1\le l_2} 2n(n-1)-n^2+l_1^2+l_2^2+l_3^2+l_1(l_3-1)    
\end{equation*}
and 
\begin{equation*}
\Delta_4\hspace{-1mm}=\hspace{-1mm}\min_{\substack {l_1+l_2+l_3+l_4=n\\ l_1\le l_2, ~l_3\le l_4}} l_1(l_2+1)-l_2(n-l_2)-l_3(l_4+1)-l_4(n-l_4).
\end{equation*}
Furthermore, numerical experiments indicate that $\Gamma^{Nor}_{Ave}\ge  \frac{\Delta_3}{2nk}=\frac{\min\{\Delta_3, \Delta_4\}}{2nk}$ for $n\in [4: 51)$ and
\begin{equation}\label{Eqn_Acess>0.72}
    \Gamma^{Nor}_{Ave}> 0.72.
\end{equation}

In \cite{zhang2025optimal}, the lower bounds on average normalized repair bandwidth and rebuilding access of $(n=k+2, k)$ MDS array codes with a sub-packetization level of $2$ are derived for the non-DRF case, which shows
\begin{equation}\label{Eqn_RB_bound_no}
\gamma^{Nor}_{Ave}\ge \frac{5}{8}  
\end{equation}
and
\begin{equation}\label{Eqn_Acess_bound_no}
\Gamma^{Nor}_{Ave}\ge \frac{4k+1}{6k}.   
\end{equation}

By \eqref{Eqn_Acess>0.72} and \eqref{Eqn_Acess_bound_no}, we note that by sacrificing the DRF property, the rebuilding access of $(k+2, k)$ MDS array codes with a sub-packetization level of $2$ can be smaller than that of the DRF case.

\section{A New Construction of $(k+2, k)$ MDS Array Codes With Asymptotically Optimal  Repair Bandwidth}\label{sec:new}

\begin{Construction}\label{Con1}
Let $w$ be a primitive element of the finite field $\mathbf{F}_q$ and $a,b\in \mathbf{F}_q\backslash\{0,1\}$. We construct an $(n=4m,k=4m-2)$ array code with sub-packetization level $2$ over $\mathbf{F}_q$, where the parity-check matrix is defined by $A$ in \eqref{Eqn parity check eq} and
\begin{equation*}
A_{4i+j}=\left\{
\begin{array}{ll}
w^i \begin{pmatrix}
1 & 0\\
0& 1
\end{pmatrix}, &\mbox{\ if\ } j=0, \\[12pt]
w^i\begin{pmatrix}
a & 0\\
0& b
\end{pmatrix}, & \mbox{\ if\ } j=1,
\\[12pt]
w^i \begin{pmatrix}
b & 0\\
a& a
\end{pmatrix}, & \mbox{\ if\ } j=2,
\\[12pt]
w^i\begin{pmatrix}
b & b\\
0& a
\end{pmatrix}, & \mbox{\ if\ }j=3,
\end{array}
\right.
\end{equation*}
for $i\in [0: m)$.
The repair matrices are defined by
\begin{equation*}
R_{4i+j}=\left\{
\begin{array}{ll}
\begin{pmatrix}
1 & a&&\\
&&1& 1
\end{pmatrix}, &\mbox{\ if\ } j=0, \\[12pt]
\begin{pmatrix}
a &1&&\\
&&1& b
\end{pmatrix}, & \mbox{\ if\ } j=1,
\\[12pt]
\begin{pmatrix}
0 & 1&&\\
&&0& 1
\end{pmatrix}, & \mbox{\ if\ } j=2,
\\[12pt]
\begin{pmatrix}
1& 0&&\\
&&1& 0
\end{pmatrix}, & \mbox{\ if\ }j=3,
\end{array}
\right.
\end{equation*}
for $i\in [0: m)$.
\end{Construction}

\begin{Theorem}\label{Thm_RB_condition}
The average normalized repair bandwidth and average normalized rebuilding access of the array code in Construction \ref{Con1} are
\begin{equation*}
\gamma^{Nor}_{Ave}=\frac{5n-8}{8n-16} \mbox{\ and\ }\Gamma^{Nor}_{Ave}=\frac{13n-16}{16n-32}
\end{equation*}
if the array code is constructed over the finite field $\mathbf{F}_q$ with $q=2^{2t}$ for some positive integer $t$, and $a=w^{\frac{q-1}{3}}$, $b=w^{\frac{2(q-1)}{3}}$. This implies that the average normalized repair bandwidth is asymptotically optimal w.r.t. the lower bound in \eqref{Eqn_RB_bound_no}.
\end{Theorem}
\begin{proof}
Consider the repair of node $4i+j$, where $i\in [0: m)$ and $j\in [0: 4)$.
\begin{itemize}
\item [i)] When $j=0$, for $i, i'\in [0: m)$, we have

\begin{align*}
R_{4i}\begin{pmatrix}
I\\
A_{4i'}
\end{pmatrix}&=\begin{pmatrix}
1 &a &&\\
& &1& 1
\end{pmatrix}\begin{pmatrix}
1 & 0\\
0& 1\\
w^{i'} & 0\\
0& w^{i'}
\end{pmatrix}\\
&=\begin{pmatrix}
1 & a \\
w^{i'}&w^{i'}
\end{pmatrix},
\end{align*}
which is of full rank if and only if
\begin{equation}\label{Eqn_a ne 1}
a\ne 1.
\end{equation}

\begin{align*}
R_{4i}\begin{pmatrix}
I\\
A_{4i'+1}
\end{pmatrix}&=\begin{pmatrix}
1 &a &&\\
& &1& 1
\end{pmatrix}\begin{pmatrix}
1 & 0\\
0& 1\\
w^{i'}a & 0\\
0& w^{i'}b
\end{pmatrix}\\&=\begin{pmatrix}
1 & a \\
w^{i'}a&w^{i'}b
\end{pmatrix},
\end{align*}
which has rank $1$ if and only if
\begin{equation}\label{Eqn_b=aa}
b=a^2.
\end{equation}

\begin{align*}
R_{4i}\begin{pmatrix}
I\\
A_{4i'+2}
\end{pmatrix}&=\begin{pmatrix}
1 &a &&\\
& &1& 1
\end{pmatrix}\begin{pmatrix}
1 & 0\\
0& 1\\
w^{i'}b & 0\\
w^{i'}a& w^{i'}a
\end{pmatrix}\\&=\begin{pmatrix}
1 & a \\
w^{i'}(a+b)&w^{i'}a
\end{pmatrix},
\end{align*}
which has rank $1$ if and only if
\begin{equation}\label{Eqn_a+b=1}
a+b=1
\end{equation}

\begin{align*}
R_{4i}\begin{pmatrix}
I\\
A_{4i'+3}
\end{pmatrix}&=\begin{pmatrix}
1 &a &&\\
& &1& 1
\end{pmatrix}\begin{pmatrix}
1 & 0\\
0& 1\\
w^{i'}b &w^{i'}b\\
0& w^{i'}a
\end{pmatrix}\\&=\begin{pmatrix}
1 & a \\
w^{i'}b&w^{i'}(a+b)
\end{pmatrix},
\end{align*}
which has rank $1$ if and only if
\begin{equation}\label{Eqn_ab=a+b}
ab=a+b.
\end{equation}

Therefore, if
\eqref{Eqn_a ne 1}--\eqref{Eqn_ab=a+b} are satisfied, the repair bandwidth and rebuilding access of node $4i$ are 
\begin{align*}
\gamma_{4i}&=\sum_{j=0,j\ne 4i}^{n-1}\mbox{rank}\left(R_{4i}\left(
\begin{array}{c}
I \\
A_{j}
\end{array}
\right)\right)\\
&=\sum_{i'=0,i'\ne i}^{m-1}\mbox{rank}\left(R_{4i}\left(
\begin{array}{c}
I \\
A_{4i'}
\end{array}
\right)\right)\\ &\quad+\sum_{i'=0}^{m-1}\sum_{j'=1}^{3}\mbox{rank}\left(R_{4i}\left(
\begin{array}{c}
I \\
A_{4i'+j'}
\end{array}
\right)\right)\\
&=2(m-1)+3m\\
&=5m-2
\end{align*}
and
\begin{equation*}
\Gamma_{4i}=\sum_{j=0,j\ne 4i}^{n-1}N_c\left(R_{4i}\left(
\begin{array}{c}
I \\
A_{j}
\end{array}
\right)\right)=8m-2.
\end{equation*}

\item [ii)] When $j=1$, for $i, i'\in [0: m)$, we have

\begin{align*}
R_{4i+1}\begin{pmatrix}
I\\
A_{4i'}
\end{pmatrix}&=\begin{pmatrix}
a &1 &&\\
& &1& b
\end{pmatrix}\begin{pmatrix}
1 & 0\\
0& 1\\
w^{i'} & 0\\
0& w^{i'}
\end{pmatrix}\\
&=\begin{pmatrix}
a & 1 \\
w^{i'}&w^{i'}b
\end{pmatrix},
\end{align*}
which has rank 1 if and only if
\begin{equation}\label{Eqn_ab=1}
ab=1.
\end{equation}

\begin{align*}
R_{4i+1}\begin{pmatrix}
I\\
A_{4i'+1}
\end{pmatrix}&=\begin{pmatrix}
a &1 &&\\
& &1& b
\end{pmatrix}\begin{pmatrix}
1 & 0\\
0& 1\\
w^{i'}a & 0\\
0& w^{i'}b
\end{pmatrix}\\&=\begin{pmatrix}
a & 1 \\
w^{i'}a&w^{i'}b^2
\end{pmatrix},
\end{align*}
which is of full rank if and only if
\begin{equation}\label{Eqn_bb ne 1}
b^2\ne 1.
\end{equation}

\begin{align*}
R_{4i+1}\begin{pmatrix}
I\\
A_{4i'+2}
\end{pmatrix}&=\begin{pmatrix}
a &1 &&\\
& &1& b
\end{pmatrix}\begin{pmatrix}
1 & 0\\
0& 1\\
w^{i'}b & 0\\
w^{i'}a& w^{i'}a
\end{pmatrix}\\&=\begin{pmatrix}
a & 1 \\
w^{i'}(b+ab)&w^{i'}ab
\end{pmatrix},
\end{align*}
which has rank $1$ if and only if
\begin{equation}\label{Eqn_aa=a+1}
a^2=a+1.
\end{equation}

\begin{align*}
R_{4i+1}\begin{pmatrix}
I\\
A_{4i'+3}
\end{pmatrix}&=\begin{pmatrix}
a &1 &&\\
& &1& b
\end{pmatrix}\begin{pmatrix}
1 & 0\\
0& 1\\
w^{i'}b &w^{i'}b\\
0& w^{i'}a
\end{pmatrix}\\&=\begin{pmatrix}
a & 1 \\
w^ib&w^i(b+ab)
\end{pmatrix},
\end{align*}
which has rank $1$ if and only if
\begin{equation}\label{Eqn_aa+a=1}
a^2+a=1
\end{equation}
holds.

Similarly, if
\eqref{Eqn_ab=1}--\eqref{Eqn_aa+a=1} are satisfied, the repair bandwidth and rebuilding access of node $4i+1$ are respectively
\begin{align*}
\gamma_{4i+1}=5m-2
\end{align*}
and
\begin{equation*}
\Gamma_{4i+1}=8m-2.
\end{equation*}

\item [iii)] When $j=2,3$, for $i,i'\in [0: m)$, it is easy to obtain

\begin{align*}
N_c\left(R_{4i+j}\begin{pmatrix}
I\\
A_{4i'+j'}
\end{pmatrix}\right)
&=\mbox{rank}\left(R_{4i+j}\begin{pmatrix}
I\\
A_{4i'+j'}
\end{pmatrix}\right)\\&=\left\{
\begin{array}{ll}
1, &\mbox{\ if\ } j'\ne j, \\[2pt]
2, & \mbox{\ if\ } j'=j.
\end{array}
\right.
\end{align*}
Thus, the repair bandwidth and rebuilding access of node $4i+j$ are
\begin{align*}
\gamma_{4i+j}=\Gamma_{4i+j}=5m-2, j=2,3.
\end{align*}
\end{itemize}

Note that if \eqref{Eqn_a ne 1}--\eqref{Eqn_aa+a=1} are satisfied, then $2\mid q$, $a\ne1$, $a^3=1$, $b=a^2$, and $a^2+a=1$.
These conditions can be satisfied if $a$ with $a\ne 1$ is a cubic root of unity in $\mathbf{F}_{2^{2t}}$ for some positive integer $t$. Specifically, \eqref{Eqn_a ne 1}--\eqref{Eqn_aa+a=1} hold if $a=w^{\frac{q-1}{3}}$ and $b=w^{\frac{2(q-1)}{3}}$ and $q=2^{2t}$.

With the above repair strategy, a direct calculation shows that
$\gamma^{Nor}_{Ave}=\frac{5n-8}{8n-16}$ and $\Gamma^{Nor}_{Ave}=\frac{13n-16}{16n-32}$, where $\gamma^{Nor}_{Ave}=\frac{5n-8}{8n-16}$ is asymptotically optimal w.r.t. the general (non-DRF-restricted) lower bound on average repair bandwidth given in \eqref{Eqn_RB_bound_no}.
\end{proof}

\begin{Theorem}\label{Thm-MDS-condition}
The array code in Construction \ref{Con1} is MDS if
\begin{itemize}
\item [i)] $a,b\not\in \{w^{t}|t\in [0:m)\cup [q-m:q-1)\}$;
\item [ii)] $a\ne w^tb$ for $t\in [-m+1: m)$;
\item [iii)] $w^{2i}+w^{2i'}-w^iw^{i'}\ne 0$ for $i,i'\in [0: m)$.
\end{itemize}
\end{Theorem}
\begin{proof}
It suffices to analyze the nonsingularity of $A_{4i+j}-A_{4i'+j'}$ for
$i,i'\in [0: m)$ and $j,j'\in [0: 4)$ with $(i, j)\ne (i', j')$, which can be carried out according to the following four cases.
\begin{itemize}
\item [i)] When $j=j'$ and $i\ne i'$,
\begin{align*}
\det(A_{4i+j}-A_{4i'+j'})&= (1-w^{i'-i})\det(A_{4i+j})\ne 0.
\end{align*}

\item [ii)] When $j=0,j'=1,2,3$,
\begin{align*}
\det(A_{4i+j}-A_{4i'+j'})=(w^i-w^{i'}a)(w^i-w^{i'}b),
\end{align*}
which is nonzero if $a,b\ne w^{i-i'}$, i.e., $a,b\not\in \{w^{t}|t\in [0:m)\cup [q-m:q-1)\}$.

\item [iii)] When $j=1,j'=2,3$,
\begin{equation*}
\det(A_{4i+j}-A_{4i'+j'})=(w^ia-w^{i'}b)(w^ib-w^{i'}a),
\end{equation*}
which is nonzero if $a\ne w^{i'-i}b, w^{i-i'}b$, i.e., $a\ne w^tb$ for $t\in [-m+1: m)$.

\item [iv)] When $j=2,j'=3$,
\begin{align*}
\det(A_{4i+j}-A_{4i'+j'})&=\hspace{-1mm} \det \begin{pmatrix}
w^ib-w^{i'}b & -w^{i'}b\\
w^ia& w^ia-w^{i'}a
\end{pmatrix}\\&= ab \det \begin{pmatrix}
w^i-w^{i'} & -w^{i'}\\
w^i& w^i-w^{i'}
\end{pmatrix}\\&= ab(w^{2i}+w^{2i'}-w^iw^{i'}),
\end{align*}
which is nonzero if $a, b\ne 0$ and $w^{2i}+w^{2i'}-w^iw^{i'}\ne 0$ for $i,i'\in [0: m)$.
\end{itemize}
\end{proof}

\begin{Theorem}\label{Thm_MDS_RB}
Let $q = 2^{2t}$ for some positive integer $t$ and $q \ge 3m+1$, then the $(n=4m, k=4m-2)$ array code in Construction~1 over $\mathbf{F}_q$ is an MDS array code with \[
\gamma_{\text{Ave}}^{\text{Nor}} = \frac{5n-8}{8n-16}  {\rm ~~and~~}
\Gamma_{\text{Ave}}^{\text{Nor}} = \frac{13n-16}{16n-32}.
\]

\end{Theorem}

\begin{proof}
By Theorems \ref{Thm_RB_condition} and \ref{Thm-MDS-condition},
it suffices to prove that i)--iii) of Theorem \ref{Thm-MDS-condition} hold under $a=w^{\frac{q-1}{3}}$ and $b=w^{\frac{2(q-1)}{3}}$.

If $q\ge 3m+1$, then we have,
\begin{equation*}
m\le \frac{q-1}{3}<q-m, m\le \frac{2(q-1)}{3}<q-m,
\end{equation*}
which implies that Theorem \ref{Thm-MDS-condition}-i) is satisfied.

Clearly, Theorem \ref{Thm-MDS-condition}-ii) holds since $a=w^{\frac{q-1}{3}}$ and $b=w^{\frac{2(q-1)}{3}}$.

Now we verify Theorem \ref{Thm-MDS-condition}-iii). For $0\le i\le i'<m$,
\begin{equation*}
w^{2i}+w^{2i'}-w^iw^{i'}=w^{2i}(1+w^{2(i'-i)}+w^{i'-i}),
\end{equation*}
which is nonzero by Lemma \ref{lem} since $0\le i'-i<m\le \frac{q-1}{3}$. The statement holds similarly for $0\le i'<i<m$. Thus, Theorem \ref{Thm-MDS-condition}-iii) is satisfied.
\end{proof}

\section{A Construction of MDS Array Codes With Two Repair Strategies}\label{sec:C2}
In this section, we present another construction of $(n=k+2, k)$ MDS array codes that incorporates two repair strategies.

\subsection{Code Construction and the MDS Property}
\begin{Construction}\label{Con2}
Let $w$ be a primitive element of the finite field $\mathbf{F}_q$, we construct an $(n=3m,k=3m-2)$ array code with a sub-packetization level of $2$ over $\mathbf{F}_q$. The parity-check matrix is defined by $A$ in \eqref{Eqn parity check eq} and
\begin{equation}\label{Eqn_Con2_blocks}
A_{3i+j}=\left\{
\begin{array}{ll}
\begin{pmatrix}
\lambda_i & -1\\
0& \lambda_i+1
\end{pmatrix}, &\mbox{\ if\ } j=0, \\[12pt]
\begin{pmatrix}
\lambda_i & 0\\
1& \lambda_i+1
\end{pmatrix}, & \mbox{\ if\ } j=1,
\\[12pt]
\begin{pmatrix}
\lambda_i+1 & 0\\
0& \lambda_i
\end{pmatrix}, & \mbox{\ if\ } j=2,
\\[12pt]
\end{array}
\right.
\end{equation}
for $i\in [0: m)$.
\end{Construction}

\begin{Theorem}\label{Thm_C2_MDS}
The $(n=3m,k=3m-2)$ array code in Construction \ref{Con2} is MDS if the following conditions are met:
\begin{itemize}
\item [i)] $\lambda_i\ne \lambda_{i'}$ for $0\le i<i'<m$;
\item [ii)] $\lambda_i-\lambda_{i'}\ne \pm1$ for $i,i'\in [0: m)$.
\end{itemize}
Furthermore, these requirements can be satisfied over a finite field $\mathbf{F}_q$ with
\begin{equation*}
q\ge\left\{
\begin{array}{ll}
\frac{2n}{3}, & \mbox{\ if\ } 2\mid q,
\\[2pt]
n, & \mbox{\ otherwise}.
\end{array}
\right.
\end{equation*}
\end{Theorem}

\begin{proof}
It suffices to analyze the nonsingularity of $A_{3i+j}-A_{3i'+j'}$ for
$i,i'\in [0: m)$ and $j,j'\in [0: 3)$ with $(i, j)\ne (i', j')$, which can be proceeded according to the following four cases.
\begin{itemize}
\item []Case 1: When $j=j'$ and $i\ne i'$, if $j=0$,
\begin{align*}
\det(A_{3i+j}-A_{3i'+j'})&= \det\begin{pmatrix}
\lambda_i-\lambda_{i'} & 0\\
0& \lambda_i-\lambda_{i'}
\end{pmatrix}\ne 0
\end{align*}
if and only if $\lambda_i\ne \lambda_{i'}$. The same argument applies for $j=1,2$.

\item []Case 2: When $j=0,j'=1$,
\begin{align*}
\det(A_{3i+j}-A_{3i'+j'})&= \det \begin{pmatrix}
\lambda_i-\lambda_{i'} & -1\\
-1& \lambda_i-\lambda_{i'}
\end{pmatrix}\\&=(\lambda_i-\lambda_{i'})^2-1\\&\ne 0,
\end{align*}
which is nonzero if $\lambda_i-\lambda_{i'}\ne \pm1$.

\item []Case 3: When $j=0,j'=2$,
\begin{align*}
\det(A_{3i+j}-A_{3i'+j'})&\hspace{-1mm}=\hspace{-1mm}\det \hspace{-1mm}\begin{pmatrix}\hspace{-1mm}
\lambda_i-\lambda_{i'}-1 & -1\\
0& \lambda_i-\lambda_{i'}+1
\hspace{-1mm}\end{pmatrix}\\&=(\lambda_i-\lambda_{i'})^2-1,
\end{align*}
which is nonzero if $\lambda_i-\lambda_{i'}\ne \pm1$.

\item []Case 4: When $j=1,j'=2$,
\begin{align*}
\det(A_{3i+j}-A_{3i'+j'})&\hspace{-1mm}=\hspace{-1mm}\det\hspace{-1mm} \begin{pmatrix}\hspace{-1mm}
\lambda_i-\lambda_{i'}-1 & 0\\
1& \lambda_i-\lambda_{i'}+1
\hspace{-1mm}\end{pmatrix}\\&=(\lambda_i-\lambda_{i'})^2-1,
\end{align*}
which is nonzero if $\lambda_i-\lambda_{i'}\ne \pm1$.
\end{itemize}
Now, let us determine the required finite field size. If the underlying finite field has characteristic two, then a finite field with size $q\ge \frac{2n}{3}=2m$ is sufficient. Since we can partition the finite field into two disjoint sets $S$ and $T$ such that for any $a\in S$ we have $a+1\in T$. We can then select pairwise distinct elements $\lambda_0, \lambda_1, \ldots, \lambda_{m-1}$ from the set $S$ to fulfill conditions i) and ii).

If the underlying finite field has an odd characteristic, a field size $q\ge n=3m$ is sufficient. Initially, we can choose $\lambda_0$ to be an arbitrary element in $\mathbf{F}_q$. For $1\le i<m$, we can select $\lambda_i$ to be any element in
\begin{equation*}
\mathbf{F}_q\backslash\{\lambda_{j},\lambda_{j}-1,\lambda_{j}+1|0\le j<i\}.
\end{equation*}
With these assignments, i) and ii) are fulfilled.
\end{proof}

\subsection{Repair Mechanisms for the Code in Construction \ref{Con2}}

In this subsection, we propose repair strategies for the code in Construction \ref{Con2}. In particular, we provide two repair strategies for nodes with indices of the form \(3i+2, i\in [0: m)\). One strategy aims to minimize the repair bandwidth, while the other focuses on minimizing the rebuilding access. It should be noted that the two strategies differ only in the repair process, not in
the underlying code structure.

We first present the repair strategy of nodes $3i$ and $3i+1$ for $i\in [0: m)$ as well as the strategy that minimizes the repair bandwidth for nodes 
$3i+2$, where $i\in [0: m)$. Under this strategy, we then demonstrate the average normalized repair bandwidth and the average normalized rebuilding access of the array code in Construction \ref{Con2}.

\begin{figure*}[!b]
\hrulefill
 \begin{align}\label{Eqn_expresssion}        
\nonumber\mathbf{f}_{3i+2}
&=-\left(R_{3i+2}\left(
\begin{array}{c}
I \\
A_{3i+2}
\end{array}
\right)\right)^{-1}\sum_{j=0, j\ne 3i+2}^{n-1}
R_{3i+2}\left(\begin{array}{c}
I \\
A_{j}
\end{array}
\right)
\mathbf{f}_j\\\nonumber&=-\begin{pmatrix}
1 & 1\\
\lambda_{i}+1 & \lambda_{i}
\end{pmatrix}^{-1}\left(\sum_{j=0}^{m-1}\begin{pmatrix}
1 & 1 \\
\lambda_{j} & \lambda_{j}
\end{pmatrix}\mathbf{f}_{3j}+\sum_{j=0}^{m-1}\begin{pmatrix}
1 & 1 \\
\lambda_{j}+1 & \lambda_{j}+1
\end{pmatrix}\mathbf{f}_{3j+1}+
\sum_{j=0, j\ne i}^{m-1}\begin{pmatrix}
1 & 1\\
\lambda_{j}+1 & \lambda_{j}
\end{pmatrix}\mathbf{f}_{3j+2}
\right)\\&=-\begin{pmatrix}
1 & 1\\
\lambda_{i}+1 & \lambda_{i}
\end{pmatrix}^{-1}\left(\sum_{j=0}^{m-1}\begin{pmatrix}
\mathbf{f}_{3j,0}+\mathbf{f}_{3j,1} \\
\lambda_{j}(\mathbf{f}_{3j,0}+\mathbf{f}_{3j,1})
\end{pmatrix}\hspace{-1mm}+\hspace{-1mm}\sum_{j=0}^{m-1}\begin{pmatrix}
\mathbf{f}_{3j+1,0}+\mathbf{f}_{3j+1,1}\\
(\lambda_{j}+1) (\mathbf{f}_{3j+1,0}+\mathbf{f}_{3j+1,1})
\end{pmatrix}\hspace{-1mm}+\hspace{-1mm}
\sum_{j=0, j\ne i}^{m-1}\begin{pmatrix}
1 & 1\\
\lambda_{j}+1 & \lambda_{j}
\end{pmatrix}\mathbf{f}_{3j+2}
\right).
    \end{align}
\end{figure*}
\begin{Theorem}[Under the first repair strategy]\label{C2_1st_repair}
The nodes of Construction \ref{Con2} can be repaired with
an average normalized repair bandwidth of $\gamma^{Nor}_{Ave}=\frac{2n-3}{3n-6}$ and an average normalized rebuilding access of $\Gamma^{Nor}_{Ave}=\frac{7n-9}{9n-18}$.
\end{Theorem}
\begin{proof}
Let the repair matrices be
\begin{equation*}
R_{3i+j}=\left\{
\begin{array}{ll}
\begin{pmatrix}
1 &0 &0&0\\
0&0&1&0
\end{pmatrix}, &\mbox{\ if\ } j=0, \\[12pt]
\begin{pmatrix}
0 &1&0&0\\
0&0&0& 1
\end{pmatrix}, & \mbox{\ if\ } j=1,
\\[12pt]
\begin{pmatrix}
1 & 1&0&0\\
0&0&1& 1
\end{pmatrix}, & \mbox{\ if\ } j=2,
\\[12pt]
\end{array}
\right.
\end{equation*}
for $i\in [0: m)$.

\begin{itemize}
\item [i)] When $j=0$, for $i, i'\in [0: m)$, we have

\begin{align*}
R_{3i}\begin{pmatrix}
I\\
A_{3i'}
\end{pmatrix}&=\begin{pmatrix}
1 &0 &0&0\\
0&0&1&0
\end{pmatrix}\begin{pmatrix}
1 & 0\\
0& 1\\
\lambda_{i'} & -1\\
0& \lambda_{i'}+1
\end{pmatrix}\\
&=\begin{pmatrix}
1 & 0 \\
\lambda_{i'} & -1
\end{pmatrix},
\end{align*}
which is of full rank.

\begin{align*}
R_{3i}\begin{pmatrix}
I\\
A_{3i'+1}
\end{pmatrix}&=\begin{pmatrix}
1 &0 &0&0\\
0&0&1&0
\end{pmatrix}\begin{pmatrix}
1 & 0\\
0& 1\\
\lambda_{i'} & 0\\
1& \lambda_{i'}+1
\end{pmatrix}\\&=\begin{pmatrix}
1 & 0 \\
\lambda_{i'} & 0
\end{pmatrix},
\end{align*}
which has rank $1$.

\begin{align*}
R_{3i}\begin{pmatrix}
I\\
A_{3i'+2}
\end{pmatrix}&=\begin{pmatrix}
1 &0 &0&0\\
0&0&1&0
\end{pmatrix}\begin{pmatrix}
1 & 0\\
0& 1\\
\lambda_{i'}+1 & 0\\
0& \lambda_{i'}
\end{pmatrix}\\&=\begin{pmatrix}
1 & 0\\
\lambda_{i'}+1 & 0
\end{pmatrix},
\end{align*}
which has rank $1$.

Thus, the repair bandwidth of node $3i$ is
\begin{align*}
\gamma_{3i}&=\sum_{j=0,j\ne 3i}^{n-1}\mbox{rank}\left(R_{3i}\left(
\begin{array}{c}
I \\
A_{j}
\end{array}
\right)\right)\\
&=\sum_{i'=0,i'\ne i}^{m-1}\mbox{rank}\left(R_{3i}\left(
\begin{array}{c}
I \\
A_{3i'}
\end{array}
\right)\right)\\&\quad+\sum_{i'=0}^{m-1}\sum_{j'=1}^{2}\mbox{rank}\left(R_{3i}\left(
\begin{array}{c}
I \\
A_{3i'+j'}
\end{array}
\right)\right)\\
&=2(m-1)+2m\\
&=4m-2,
\end{align*}
while
the rebuilding access is
\begin{align*}
\Gamma_{3i}&=\sum_{j=0,j\ne 3i}^{n-1}N_c\left(R_{3i}\left(
\begin{array}{c}
I \\
A_{j}
\end{array}
\right)\right)\\
&=\sum_{i'=0,i'\ne i}^{m-1}N_c\left(R_{3i}\left(
\begin{array}{c}
I \\
A_{3i'}
\end{array}
\right)\right)\\&\quad+\sum_{i'=0}^{m-1}\sum_{j'=1}^{2}N_c\left(R_{3i}\left(
\begin{array}{c}
I \\
A_{3i'+j'}
\end{array}
\right)\right)\\
&=2(m-1)+2m\\
&=4m-2,
\end{align*}

\item [ii)] When $j=1$, similar to the previous case, we have that the repair bandwidth and rebuilding access of node $3i+1$ are
\begin{equation*}
\gamma_{3i+1}=\Gamma_{3i+1}=4m-2.
\end{equation*}

\item [iii)] When $j=2$, for $i,i'\in [0: m)$, we have

\begin{align*}
R_{3i+2}\begin{pmatrix}
I\\
A_{3i'+2}
\end{pmatrix}&=\begin{pmatrix}
1 & 1&0&0\\
0&0&1& 1
\end{pmatrix}\begin{pmatrix}
1 & 0\\
0& 1\\
\lambda_{i'}+1 & 0\\
0& \lambda_{i'}
\end{pmatrix}\\&
=\begin{pmatrix}
1 & 1\\
\lambda_{i'}+1 & \lambda_{i'}
\end{pmatrix},
\end{align*}
which is of full rank.

\begin{align*}
R_{3i+2}\begin{pmatrix}
I\\
A_{3i'}
\end{pmatrix}&=\begin{pmatrix}
1 & 1&0&0\\
0&0&1& 1
\end{pmatrix}\begin{pmatrix}
1 & 0\\
0& 1\\
\lambda_{i'} & -1\\
0& \lambda_{i'}+1
\end{pmatrix}\\&=\begin{pmatrix}
1 & 1 \\
\lambda_{i'} & \lambda_{i'}
\end{pmatrix},
\end{align*}
which has rank $1$.

\begin{align*}
R_{3i+2}\begin{pmatrix}
I\\
A_{3i'+1}
\end{pmatrix}&=\begin{pmatrix}
1 & 1&0&0\\
0&0&1& 1
\end{pmatrix}\begin{pmatrix}
1 & 0\\
0& 1\\
\lambda_{i'} & 0\\
1& \lambda_{i'}+1
\end{pmatrix}\\&=\begin{pmatrix}
1 & 1\\
\lambda_{i'}+1 & \lambda_{i'}+1
\end{pmatrix},
\end{align*}
which has rank $1$.

That is, the repair bandwidth of node $3i+2$ is
\begin{align*}
\gamma_{3i+2}&=\sum_{j=0,j\ne 3i+2}^{n-1}\mbox{rank}\left(R_{3i+2}\left(
\begin{array}{c}
I \\
A_{j}
\end{array}
\right)\right)\\
&=\sum_{i'=0,i'\ne i}^{m-1}\mbox{rank}\left(R_{3i+2}\left(
\begin{array}{c}
I \\
A_{3i'+2}
\end{array}
\right)\right)\\&\quad+\sum_{i'=0}^{m-1}\sum_{j'=0}^{1}\mbox{rank}\left(R_{3i+2}\left(
\begin{array}{c}
I \\
A_{3i'+j'}
\end{array}
\right)\right)\\
&=2(m-1)+2m\\
&=4m-2.
\end{align*}
Similarly, the rebuilding access of node $3i+2$ is $\Gamma_{3i+2}=6m-2$.
\end{itemize}
With the above repair strategy, a direct calculation shows that
$\gamma^{Nor}_{Ave}=\frac{2n-3}{3n-6}$ and $\Gamma^{Nor}_{Ave}=\frac{7n-9}{9n-18}$.
\end{proof}

From the proof of Theorem \ref{C2_1st_repair}, \eqref{Eqn parity check eq2} can be rewritten as in \eqref{Eqn_expresssion}.

For completeness, we provide the repair procedure for node $3i+2$ in the following, while that for the other nodes is similar.

\textbf{Repair procedure for nodes $3i+2$, $i\in [0: m)$ (Under the first repair strategy):}
\begin{enumerate}
    \item \textbf{Transmission Phase:} 
From helper node $l$, 
\begin{itemize}
    \item If $l=3j$ or $3j+1$, node $l$ accesses all its symbols, computes $f_{l, 0}+f_{l, 1}$ locally, and then transmits it to the central coordinator  (or the new replacement node);
\item If $l=3j+2$ for some $j\in [0:m)\backslash\{i\}$,  node $l$ accesses all its symbols and transmits them to the central coordinator. 
\end{itemize}

    \item \textbf{Decode Phase:} From the downloaded data, the decoder regenerates the two unknown symbols $f_{3i+2,0}$ and $f_{3i+2,1}$ by computing \eqref{Eqn_expresssion}.
\end{enumerate}

In the following, we propose an alternative repair strategy for nodes $3i+2$ where $i\in [0: m)$, which aims to minimize the rebuilding access. We also present the average normalized repair bandwidth and average normalized rebuilding access of the array code in Construction \ref{Con2} under this repair strategy.
\begin{Theorem}[Under the second repair strategy]\label{C2_2st_repair}
The nodes of Construction \ref{Con2} can be repaired with the average normalized repair bandwidth and average normalized rebuilding access being $$\gamma'^{Nor}_{Ave}=\Gamma'^{Nor}_{Ave}=\frac{13n-18}{18n-36}.$$
\end{Theorem}
\begin{proof}
We modify the repair matrix of node $3i+2$
to
\begin{equation*}
R'_{im+2}=\begin{pmatrix}
1 & 0&0&0\\
0&1&1& 0
\end{pmatrix}
\end{equation*}
for $i\in [0: m)$, while keeping the repair matrices of the other nodes the same as those used in the first repair strategy in the proof of Theorem \ref{C2_1st_repair}.

Now for $i,i'\in [0: m)$, we have

\begin{align*}
R'_{3i+2}\begin{pmatrix}
I\\
A_{3i'+2}
\end{pmatrix}&=\begin{pmatrix}
1 & 0&0&0\\
0&1&1& 0
\end{pmatrix}\begin{pmatrix}
1 & 0\\
0& 1\\
\lambda_{i'}+1 & 0\\
0& \lambda_{i'}
\end{pmatrix}\\
&=\begin{pmatrix}
1 & 0\\
\lambda_{i'}+1 & 1
\end{pmatrix},
\end{align*}
which is of full rank.

\begin{figure*}[!b]
\hrulefill
 \begin{align}\label{Eqn_expresssion2}        
\nonumber\mathbf{f}_{3i+2}
&=-\left(R_{3i+2}\left(
\begin{array}{c}
I \\
A_{3i+2}
\end{array}
\right)\right)^{-1}\sum_{j=0, j\ne 3i+2}^{n-1}
R_{3i+2}\left(\begin{array}{c}
I \\
A_{j}
\end{array}
\right)
\mathbf{f}_j\\\nonumber&=-\begin{pmatrix}
1 & 0\\
\lambda_{i}+1 & 1
\end{pmatrix}^{-1}\left(\sum_{j=0}^{m-1}\begin{pmatrix}
1 & 0 \\
\lambda_{j} & 0
\end{pmatrix}\mathbf{f}_{3j}+\sum_{j=0}^{m-1}\begin{pmatrix}
1 & 0 \\
\lambda_{j} & 1
\end{pmatrix}\mathbf{f}_{3j+1}+
\sum_{j=0, j\ne i}^{m-1}\begin{pmatrix}
1 & 0\\
\lambda_{j}+1 & 1
\end{pmatrix}\mathbf{f}_{3j+2}
\right)\\&=-\begin{pmatrix}
1 & 0\\
\lambda_{i}+1 & \lambda_{i}
\end{pmatrix}^{-1}\left(\sum_{j=0}^{m-1}\begin{pmatrix}
\mathbf{f}_{3j,0} \\
\lambda_{j}\mathbf{f}_{3j,0}
\end{pmatrix}+\sum_{j=0}^{m-1}\begin{pmatrix}
1 & 0 \\
\lambda_{j} & 1
\end{pmatrix}\mathbf{f}_{3j+1}+
\sum_{j=0, j\ne i}^{m-1}\begin{pmatrix}
1 & 0\\
\lambda_{j}+1 & 1
\end{pmatrix}\mathbf{f}_{3j+2}
\right).
    \end{align}
\end{figure*}

\begin{align*}
R'_{3i+2}\begin{pmatrix}
I\\
A_{3i'}
\end{pmatrix}&=\begin{pmatrix}
1 & 0&0&0\\
0&1&1& 0
\end{pmatrix}\begin{pmatrix}
1 & 0\\
0& 1\\
\lambda_{i'} & -1\\
0& \lambda_{i'}+1
\end{pmatrix}\\&=\begin{pmatrix}
1 & 0 \\
\lambda_{i'} &0
\end{pmatrix},
\end{align*}
which has rank $1$.

\begin{align*}
R'_{3i+2}\begin{pmatrix}
I\\
A_{3i'+1}
\end{pmatrix}&=\begin{pmatrix}
1 & 0&0&0\\
0&1&1& 0
\end{pmatrix}\begin{pmatrix}
1 & 0\\
0& 1\\
\lambda_{i'} & 0\\
1& \lambda_{i'}+1
\end{pmatrix}\\&=\begin{pmatrix}
1 & 0\\
\lambda_{i'} & 1
\end{pmatrix},
\end{align*}
which has a rank of $2$.

With this repair strategy, the repair bandwidth of node $3i+2$ is
\begin{align*}
\gamma'_{3i+2}&=\sum_{j=0,j\ne 3i+2}^{n-1}\mbox{rank}\left(R'_{3i+2}\left(
\begin{array}{c}
I \\
A_{j}
\end{array}
\right)\right)\\
&=\sum_{i'=0,i'\ne i}^{m-1}\mbox{rank}\left(R'_{3i+2}\left(
\begin{array}{c}
I \\
A_{3i'+2}
\end{array}
\right)\right)\\&\quad+\sum_{i'=0}^{m-1}\sum_{j'=0}^{1}\mbox{rank}\left(R'_{3i+2}\left(
\begin{array}{c}
I \\
A_{3i'+j'}
\end{array}
\right)\right)\\
&=2(m-1)+3m\\
&=5m-2.
\end{align*}
Similarly, the rebuilding access of node $3i+2$ is $\Gamma'_{3i+2}=5m-2$.

With the new repair strategy for nodes $3i+2$ where $i\in [0: m)$, a direct calculation shows that
$\gamma'^{Nor}_{Ave}=\Gamma'^{Nor}_{Ave}=\frac{13n-18}{18n-36}$.
\end{proof}

From the proof of Theorem \ref{C2_2st_repair}, \eqref{Eqn parity check eq2} can be rewritten as in \eqref{Eqn_expresssion2}.

\textbf{Repair procedure for nodes $3i+2$, $i\in [0: m)$ (Under second repair strategy):}
\begin{enumerate}
    \item \textbf{Transmission Phase:} 
From helper node $l$, 
\begin{itemize}
\item If $l=3j$, node $l$ accesses the symbol $f_{l, 0}$ and transmits it to the central coordinator;
\item If $l=3j+2$ for some $j\in [0:m)\backslash\{i\}$ or $l=3i+1$ for some $j\in [0:m)$, node $l$ accesses all its symbols and transmits them to the central coordinator. 
\end{itemize}

\item \textbf{Decode Phase:} From the downloaded data, the decoder regenerates the two unknown symbols $f_{3i+2,0}$ and $f_{3i+2,1}$ by computing \eqref{Eqn_expresssion2}.
\end{enumerate}

\begin{Remark}\label{remark1}
In Construction~\ref{Con2}, we assumed $n = 3m$.
The parity-check matrix is built from three types of matrices: upper triangular, lower triangular, and diagonal.
Each type appears an equal number of times. While this construction is systematic, it does not always yield the smallest average rebuilding access. An extreme optimization of the average rebuilding access may dynamically adjust the number of occurrences of each type of building block. For example, we can  replace the building blocks of the parity-check matrix
in \eqref{Eqn_Con2_blocks} with the following

\begin{equation*}
A_{i}=\left\{
\begin{array}{ll}
\begin{pmatrix}
\lambda_i & -1\\
0& \lambda_i+1
\end{pmatrix}, &\mbox{\ if\ } 0\le i<l_1, \\[12pt]
\begin{pmatrix}
\lambda_{i-l_1} & 0\\
1& \lambda_{i-l_1}+1
\end{pmatrix}, & \mbox{\ if\ } l_1\le i<l_1+l_2,
\\[12pt]
\begin{pmatrix}
\lambda_{i-l_1-l_2}+1 & 0\\
0& \lambda_{i-l_1-l_2}
\end{pmatrix}, & \mbox{\ if\ } l_1+l_2\le i<n.
\\[12pt]
\end{array}
\right.
\end{equation*}
\end{Remark}

\begin{Remark}
If we replace the building blocks of the parity-check matrix
in \eqref{Eqn_Con2_blocks} according to Remark \ref{remark1}, and
let $l_3=n-l_1-l_2$, then a direct calculation shows that
\begin{align}\label{Eqn_alt}
\Gamma^{Nor}_{Ave}=\frac{n^2-2n+l_1^2+l_2^2+l_3^2+l_2l_3}{2nk},
\end{align}
and we can choose the values of  $l_1$ and $l_2$ such that $l_3\le l_1$ and $\Gamma^{Nor}_{Ave}$ in \eqref{Eqn_alt} is minimized.
In addition, the repair of only the last $l_3$ nodes requires access to one more symbol than the construction in \cite{wu2021achievable} that achieves the lower bound in \eqref{Eqn_Acess_bound}.

In fact, $0.7222<\lim\limits_{n\to +\infty}\frac{13n-18}{18n-36}=\frac{13}{18}<0.7223$, indicating that the average normalized rebuilding access of the MDS code in Construction \ref{Con2} is very close to the lower bounds in \eqref{Eqn_Acess_bound} and \eqref{Eqn_Acess>0.72}.
\end{Remark}

\section{Comparisons}\label{sec:comp}
In this section, we compare some key parameters among the proposed MDS array codes and several existing notable MDS codes with two parity nodes and a sub-packetization level of $2$.

\begin{table*}[htbp]
\begin{center}
\caption{A comparison of some key parameters among the MDS array codes proposed in this paper and some existing ones with two parity nodes and a sub-packetization level of $2$}\label{Table comp}
\setlength{\tabcolsep}{2.7pt}
\begin{tabular}{|c|c|c|c|c|c|c|}
\hline
& Code length $n$ & $\gamma^{Nor}_{Ave}$ & $\Gamma^{Nor}_{Ave}$ & Required field size $q$& References\\
\hline
\hline
The product-matrix MSR code & $5$ & $2/3$ & $4/3$ & $\ge 10$& \cite{rashmi2011optimal}\\
\hline
The long MDS code & $6$ & $3/4$ & $23/24$ & $\ge 4$ & \cite{wang2016explicit}\\
\hline
The GKW code & $5$ & $2/3$ &$7/6$ & $\ge 4$& \cite{guan2017construction}\\
\hline
The GKWX code & $6$ & $5/8$ &$47/48$& $\ge 4$ & \cite{guan2017new}\\
\hline
The WHLBZZW code & Arbitrary & $>0.72$ &$>0.72$& Sufficient large & \cite{wu2021achievable}\\
\hline
Non-DRF code 1  & Arbitrary & $\frac{k+\lfloor n/4\rfloor+\frac{1}{n}(n\mod 4)\lceil n/4\rceil}{2k}$ &Not discussed& $>n+2$ & \cite{zhang2025optimal}\\[3pt]
\hline
Non-DRF Code 2& Arbitrary & Not discussed &$\frac{k+\lfloor n/3\rfloor+\frac{1}{n}(n\mod 3)\lceil n/3\rceil}{2k}$& $>n$ & \cite{zhang2025optimal}\\[3pt]
\hline
LDIO Code & Arbitrary & $\frac{3k+8}{4k+8}$& $\frac{3k+8}{4k+8}$& ${\rm Sufficiently ~large}$ & \cite{li2026mds}\\[3pt]
\hline
MDS codes in Construction \ref{Con1} & Arbitrary & $\frac{5n-8}{8n-16}$ (Asymptotically optimal) &$\frac{13n-16}{16n-32}$ & $>\frac{3n}{4}$ & Thm. \ref{Thm_MDS_RB}\\[3pt]
\hline
MDS codes in Construction \ref{Con2} & Arbitrary & $\left\{\hspace{-2mm}
\begin{array}{ll}
\frac{2n-3}{3n-6},&\hspace{-2mm} \mbox{1st\ repair\ strategy}
\\[4pt]
\frac{13n-18}{18n-36},&\hspace{-2mm} \mbox{2nd\ repair\ strategy}
\end{array}
\right. $ & $\left\{\hspace{-2mm}
\begin{array}{ll}
\frac{7n-9}{9n-18}, &\hspace{-2mm} \mbox{1st\ repair\ strategy}
\\[4pt]
\frac{13n-18}{18n-36}, &\hspace{-2mm} \mbox{2nd\ repair\ strategy}
\end{array}
\right.$ &$\ge\left\{\hspace{-2mm}
\begin{array}{ll}
\frac{2n}{3}, &\hspace{-2mm} \mbox{if\ } 2\mid q
\\[4pt]
n, & \hspace{-2mm}\mbox{otherwise}
\end{array}
\right.$ & Thms. \ref{Thm_C2_MDS}-\ref{C2_2st_repair}\\
\hline
\end{tabular}
\end{center}
\end{table*}

Table \ref{Table comp} includes existing MDS array codes that are related to this work. Notably, the WHLBZZW code in \cite{wu2021achievable} is optimal w.r.t. the lower bound on average rebuilding access given in \eqref{Eqn_Acess_bound}, the Non-DRF code 1 in \cite{zhang2025optimal} is asymptotically optimal w.r.t. the general (non-DRF-restricted) lower bound on the average repair bandwidth in \eqref{Eqn_RB_bound_no} and the Non-DRF code 2 is asymptotically optimal w.r.t. the general (non-DRF-restricted) lower bound on the average rebuilding access in \eqref{Eqn_Acess_bound_no}. The LDIO code in \cite{li2026mds} works for general code parameters as long as the number $r$ of parity nodes and the sub-packetization level $N$ satisfy $(N-1)\mid (r-1)$.

Before presenting the comparison, we clarify the range of code lengths supported by our constructions. For convenience of notation, Construction~\ref{Con1} is defined for $n = 4m$ (multiples of $4$), and Construction~\ref{Con2} for $n = 3m$ (multiples of $3$).
Note that for arbitrary $n$ that do not satisfy these divisibility conditions, one can accommodate it by constructing the code for the smallest multiple $n' \ge n$ and then shortening (removing $n'-n$ nodes) \cite{sasidharan2016explicit,macwilliams1977theory}. 

From Table \ref{Table comp}, we observe that the proposed MDS codes have the following advantages:
\begin{itemize}
\item The product-matrix MSR code in \cite{rashmi2011optimal}, the long MDS code in \cite{wang2016explicit}, the GKW code in \cite{guan2017construction}, and the GKWX code in \cite{guan2017new} with two parities and a sub-packetization level of $2$ only support a code length of $5$ or $6$. In contrast, the new MDS array codes in Constructions \ref{Con1} and \ref{Con2} support arbitrary code lengths.

\item Compared to the WHLBZZW code in \cite{wu2021achievable}, the new MDS code in Construction \ref{Con1} has a smaller repair bandwidth and is explicit over a small finite field by sacrificing the rebuilding access.

\item  The recent LDIO code in \cite{li2026mds} works for general code parameters but should be constructed over a sufficiently large finite field. In addition, it only focuses on optimizing the rebuilding access of systematic nodes. In contrast, under the parameters of two parity nodes, both new MDS array codes in Constructions \ref{Con1} and \ref{Con2} are explicit and require a smaller finite field size. Additionally, the new MDS array code in Construction \ref{Con1} has a smaller average repair bandwidth for all nodes, and the new MDS array code in Construction \ref{Con2} has a smaller average rebuilding access for all nodes (under the second repair strategy).


\item Compared to the Non-DRF codes in \cite{zhang2025optimal}, both new MDS array codes in Constructions \ref{Con1} and \ref{Con2} exhibit the DRF property and require a smaller finite field size, and the new MDS array code in Construction \ref{Con1} is also asymptotically optimal w.r.t. the lower bound in \eqref{Eqn_RB_bound_no}.

\item The new MDS code in Construction \ref{Con2} supports two repair mechanisms. The first mechanism achieves a smaller repair bandwidth than that in \cite{wu2021achievable}, while the second mechanism results in rebuilding access that is very close to the lower bounds in \eqref{Eqn_Acess_bound} and \eqref{Eqn_Acess>0.72}.
\end{itemize}

\section{Conclusion}\label{sec:concl}
In this paper, we proposed two constructions of MDS array codes with two parity nodes and a sub-packetization level of $2$ for arbitrary code lengths. The required finite field sizes, which are smaller than the code lengths, were determined for both constructions. The first construction offers the smallest repair bandwidth among all existing constructions with the same parameters.
It is also asymptotically optimal w.r.t. the general (non-DRF-restricted) lower bound derived in \cite{zhang2025optimal}.
The second construction supports two repair mechanisms.
One mechanism focuses on optimizing repair bandwidth, while the other emphasizes rebuilding access.
While extending these constructions to include more parities is feasible, determining the required finite field sizes presents challenges that will be addressed in future research.

\section*{Acknowledgment}
The authors would like to thank the Associate Editor
Prof. Itzhak Tamo and the four anonymous reviewers for
their valuable suggestions and comments, which have greatly
improved the presentation and quality of this paper. The first author wishes to thank Prof. Linqi Song for helpful discussions and for his enormous support during the transition to City University of Hong Kong.

\ifCLASSOPTIONcaptionsoff
\newpage
\fi

\bibliographystyle{ieeetr}
\bibliography{MDS_subpack_two}
\begin{IEEEbiographynophoto}{Jie Li} (Senior~Member, IEEE) received the B.S. and M.S. degrees in mathematics from Hubei University, Wuhan, China, in 2009 and 2012, respectively, and received the Ph.D. degree from the Department of Communication Engineering, Southwest Jiaotong University, Chengdu, China, in 2017.

From  2015 to   2016, he was a visiting Ph.D. student at the Department of Electrical Engineering and Computer Science, The University of Tennessee at Knoxville, TN, USA.  From   2017 to 2019, he was a postdoctoral researcher at the Department of Mathematics, Hubei University, Wuhan, China. From 2019 to   2021, he was a postdoctoral researcher at the Department of Mathematics and Systems Analysis, Aalto University, Finland. From 2021 to 2026, he was a Senior Researcher with the Theory Lab, Central Research Institute, 2012 Labs, Huawei Technologies Co., Ltd., Hong Kong SAR, China.  He is currently a Research Fellow with the Department of Computer Science, City University of Hong Kong, Hong Kong SAR, China. His research interests include coding for distributed storage and distributed computing, private information retrieval, and sequence design.

Dr. Li received the IEEE Jack Keil Wolf ISIT Student Paper Award and the  
Honor of Outstanding Graduates of Sichuan Province both in 2017.
\end{IEEEbiographynophoto}

\begin{IEEEbiographynophoto}{Xiaohu Tang} (Fellow, IEEE)  received the B.S. degree in applied mathematics
from Northwest Polytechnic University, Xi’an, China, in 1992, the M.S. degree in applied mathematics from Sichuan University, Chengdu, China, in 1995, and the Ph.D. degree in electronic engineering from Southwest Jiaotong
University, Chengdu, in 2001.

From 2003 to 2004, he was a Research Associate with the Department of
Electrical and Electronic Engineering, The Hong Kong University of Science
and Technology. From 2007 to 2008, he was a Visiting Professor with the
University of Ulm, Germany. Since 2001, he has been with the School of
Information Science and Technology, Southwest Jiaotong University, where he
is currently a Professor. His research interests include coding theory, network
security, distributed storage, and information processing for big data.

Dr. Tang was the recipient of the National Excellent Doctoral Dissertation
award in 2003 (China), the Humboldt Research Fellowship in 2007
(Germany), and the Outstanding Young Scientist Award by NSFC in 2013
(China). He served as Associate Editors for several journals including \textit{IEEE
Transactions on Information Theory} and \textit{IEICE Transactions on
Fundamentals}, and served on a number of technical program committees of
conferences.
\end{IEEEbiographynophoto}

\end{document}